\documentclass[journal]{IEEEtran}

\usepackage{cite,amsmath}
\usepackage{amssymb,amsfonts,amsthm,enumerate,cleveref}

\newtheorem{theorem}{Theorem}
\newtheorem{corollary}{Corollary}
\newtheorem{proposition}{Proposition}

\begin{document}

\title{Optimum Linear Codes with Support Constraints over Small Fields}

\author{Hikmet~Yildiz,~\quad~Babak~Hassibi\\
California Institute of Technology\\
Email: \{hyildiz, hassibi\}@caltech.edu}


\maketitle

\begin{abstract}
We consider the problem of designing optimal linear codes (in terms of having the largest minimum distance) subject to a support constraint on the generator matrix.
We show that the largest minimum distance can be achieved by a subcode of a Reed-Solomon code of small field size. As a by-product of this result, we settle the GM-MDS conjecture of Dau \textit{et. al.} in the affirmative.
\end{abstract}


%
%
\section{Introduction}
The problem of designing a linear code with the largest possible minimum distance, subject to support constraints on the generator matrix, has recently found several applications. These include
multiple access networks \cite{halbawi2014distributed,dau2015simple} as well as
weakly secure data exchange
\cite{yan2013algorithms,yan2014weakly}. 
A simple upper bound on the maximum minimum distance can be obtained from a sequence of Singleton bounds (see eq. (\ref{eq_singleton}) below)
and can further be achieved by randomly choosing the nonzero elements of the generator matrix from a field of a large enough size.

A natural question to ask is whether the above maximum minimum distance can be achieved with a field of small size, and in particular with a structured, possible algebraic, construction.
This question is equivalent to a recently proposed conjecture by 
Dau \textit{et al.} \cite{dau2014existence}, which is commonly referred to as the GM-MDS conjecture.

In the past couple of years, progress has been reported on this conjecture.
Heidarzadeh \textit{et al.} \cite{heidarzadeh2017algebraic} have proved it for dimensions $k\leq 5$.
Halbawi \textit{et al.} \cite{halbawi2014distributed} have
proved the statement for $m\leq 3$ if there are $m$ distinct support sets on the rows of the generator matrix.
In the authors' previous work \cite{isit18}, the statement has been proved for ${m\leq 6}$. Halbawi \textit{et al.} \cite{halbawi2016balanced,halbawi2016balanced2} and Song \textit{et al.} \cite{song2018generalized}, have studied the problem when the generator matrix is sparsest and balanced and established the conjecture in this special case. Yan \textit{et al.} \cite{yan2014weakly} give some partial results.

In this paper
we show that the largest minimum distance can be achieved by a subcode of a Reed-Solomon code of small field size, in fact as low as $2n-d$, where $n$ is the code length, and $d$ is the maximum minimum distance dictated by the support constraints. As a by-product of this result, we settle the GM-MDS conjecture in the affirmative.

The remainder of the paper is organized as follows.
In Section II,
we characterize the generator matrices of subcodes of Reed-Solomon codes.
Section III defines the problem (of maximizing $d_{\min}$ subject to support constraints) and shows that it can be reduced to the GM-MDS conjecture. Section IV proposes a more general statement of the problem, that is not directly related to the coding problem, but that more readily lends itself to an induction argument. This is the result we prove which, by fiat, solves the problem of maximizing $d_{\min}$ and the GM-MDS conjecture. The proof is detailed in the Appendix.

\subsection{Notation}
Matrices are shown by bold capital letters and vectors are shown by bold lower case letters.
For $n\geq 0$, we denote by $[n]$ the set $\{1,2,\dots,n\}$ by admitting $[0]=\emptyset$.
For $n\geq 1$, we write $[\theta_i]_{i=1}^m$ to represent the ordered list of objects $\theta_1,\dots,\theta_n$.
For a finite nonempty $S\subset\mathbb Z$, $[\theta_i]_{i\in S}$ is the ordered list of $\theta_i$'s for $i\in S$ in the ascending order of their indices.

$\mathbb{F}[x]$ represents the polynomial ring over the field $\mathbb{F}$, i.e. the set of polynomials with coefficients in $\mathbb{F}$.
$\mathbb{F}(x)$ represents the field of rational functions in $x$ over the field $\mathbb{F}$, i.e. the set of functions that can be written as a ratio of two polynomials in $\mathbb{F}[x]$ such that the denominator is not the zero polynomial.

$[n,k]_q$ and $[n,k,d]_q$ represent a linear code over $\mathbb{F}_q$ with length $n$, dimension $k$, and minimum distance $d$.

%
%
\section{Subcodes of Reed-Solomon Codes}\label{subcodes_rs}
An $[n,\ell,n-\ell+1]_q$ Reed-Solomon code can be generated by a Vandermonde matrix
\begin{equation}
\label{eq_vandermonde}
	\mathbf V=\begin{pmatrix}
		1 & 1 & \cdots & 1\\
		\alpha_1 & \alpha_2 & \cdots & \alpha_n\\
		\vdots & \vdots & & \vdots\\
		\alpha_1^{\ell-1} & \alpha_2^{\ell-1} & \cdots & \alpha_n^{\ell-1}
	\end{pmatrix}\in\mathbb F_q^{\ell\times n}
\end{equation}
for distinct $\alpha_1,\dots,\alpha_n\in\mathbb F_q$.
Reed-Solomon codes have efficient decoders that can correct up to $\lfloor\frac{n-\ell+1}2\rfloor$ errors.

For $n\geq\ell\geq k$, $[n,k]_q$ subcodes of $[n,\ell]_q$ Reed-Solomon codes have generator matrices of the following form:
\begin{equation}
	\mathbf G=\mathbf T\cdot\mathbf V
\end{equation}
where $\mathbf T\in\mathbb F_q^{k\times \ell}$ is full rank and $\mathbf V$ is given in (\ref{eq_vandermonde}).

The minimum distance $d$ is equal to the minimum weight of $\mathbf m \mathbf G$ over all nonzero row vectors $\mathbf m\in\mathbb{F}_q^k$.
Since $\mathbf V$ is a Vandermonde matrix, if we treat the entries of $\mathbf m\mathbf T$ as coefficients of a nonzero polynomial $p\in\mathbb F_q[x]$, then, the entries of $\mathbf m\mathbf G$ will be $p(\alpha_1),\dots,p(\alpha_n)$.
As $\deg p\leq l-1$, the number of nonzero entries in $\mathbf m\mathbf G$ is at least $n-\ell+1$. Therefore, the minimum distance is bounded by $d\geq n-\ell +1$.

We should mention that every $[n,k]_q$ linear code is a subcode of an $[n,n]_q$ Reed-Solomon code. However, we are interested in the subcodes of Reed-Solomon codes with the same minimum distance. In other words, we want to design $\ell$, $\mathbf T$, and $\mathbf V$ such that $d=n-\ell+1$.
Note that in that case, we can use the same decoder of the Reed-Solomon code with the generator matrix $\mathbf V$ to correct up to $\lfloor\frac d2\rfloor$ errors.

%
%
\section{Support Constraint on the Generator Matrix}
Let $S_1,S_2,\dots,S_k\subset[n]$.
For an $[n,k]_q$ linear code, suppose that we have the support constraints ${\forall i\in[k]},{\forall j\in S_i},{\mathbf G_{ij}=0}$ on the generator matrix $\textbf G$.
We are interested in finding a linear code having the largest minimum distance under these constraints.

For any nonempty $\Omega\subset[k]$, the rows of $\mathbf G$ indexed in $\Omega$ have zeros in all their entries indexed in $\bigcap_{i\in\Omega}S_i$.
Consider the submatix of $\mathbf G$ consisting of the rows indexed in $\Omega$ and the columns indexed in $[n]-\bigcap_{i\in\Omega}S_i$.
The minimum distance $d$ of $\mathbf G$ is at most the minimum distance of the code generated by this submatrix, which is at most $n-\left|\bigcap_{i\in\Omega}S_i\right|-|\Omega|+1$ by the Singleton bound. Hence,
\begin{equation}
\label{eq_singleton}
	d\leq n+1-\max_{\emptyset\neq\Omega\subset[k]}\left|\bigcap_{i\in\Omega}S_i\right|+|\Omega|
\end{equation}

\Cref{thm_subcodes} states that we can achieve this bound using the subcodes of Reed-Solomon Codes if $q$ is large enough. In the proof of \Cref{thm_subcodes}, we will show that it is enough to prove a special case, \Cref{thm_gmmds}, where ${\ell=k}$ and the optimum code is an MDS code.

\begin{theorem}
	\label{thm_subcodes}
	Let $S_1,S_2,\dots,S_k\subset [n]$,
	\begin{equation}
		\ell \triangleq \max_{\emptyset\neq\Omega\subset[k]}\left|\bigcap_{i\in\Omega}S_i\right|+|\Omega|
	\end{equation}
	and $q\geq n+\ell-1$ be a field size.
	Then, there exists an $[n,k,d]_q$ subcode of a Reed-Solomon code that achieves ${d= n-\ell+1}$ such that the generator matrix $\mathbf G$ has zeros at $(i,j)$ for ${j\in S_i},{i\in[k]}$.
    \hfill$\diamond$
\end{theorem}
\begin{proof}
	For $\Omega=[k]$, we have $\ell\geq k$.
	Let $S_{k+1},\dots,S_{\ell}=\emptyset$.
	By \Cref{thm_gmmds}, there exists an $[n,\ell]_q$ Reed-Solomon code that has a generator matrix $\mathbf G'$ such that $\mathbf G'_{ij}=0$ for $j\in S_i, i\in[\ell]$.
	The subcode with the generator matrix $\mathbf G$ consisting of the first $k$ rows of $\mathbf G'$ satisfies the desired constraints.
\end{proof}

\subsection{Existence of MDS codes}\label{section_mds}
As a special case, we will describe necessary and sufficient conditions on the support constraints for the existence of a Reed-Solomon code whose generator matrix satisfies these support constraints.
\Cref{thm_gmmds} has been known as the GM-MDS conjecture, which is proposed by Dau \textit{et al.} in \cite{dau2014existence}, where they showed that it is equivalent to \Cref{corollary1} given in \Cref{section_mainthm}.

\begin{theorem}
	\label{thm_gmmds}
	Let $S_1,S_2,\dots,S_k\subset[n]$ and $q\geq n+k-1$ be a field size.
	For any nonempty $\Omega\subset[k]$,
	\begin{equation}
	\label{condcorollary}
	\left|\bigcap_{i\in\Omega}S_i\right|\leq k-|\Omega| 
	\end{equation}
	if and only if there exists an $[n,k]_q$ Reed-Solomon code with a generator matrix $\mathbf G$ such that $\mathbf G_{ij}=0$ for $j\in S_i, {i\in[k]}$.
	\hfill$\diamond$
\end{theorem}

%
%
\section{Main Theorem}\label{section_mainthm}
For $n\geq 0$, let $\mathbb{K}_n=\mathbb{Q}(\alpha_1,\dots,\alpha_n)$ be the field of rational functions in $\alpha_1,\dots,\alpha_n$ over the set of rational numbers. We will admit that $\mathbb{K}_0=\mathbb{Q}$.
For $k\geq m\geq 1$ and $n\geq 0$, define
\begin{IEEEeqnarray}{rCl}
	\mathcal{S}_{k,m,n}
	& = & \biggl\{
	[(S_i,r_i)]_{i=1}^m 
	\biggm|
	\forall i\in[m]\: S_i\subset[n], r_i\in\mathbb{Z}^+,
	\nonumber\\
	&& \qquad\qquad
	|S_i|+r_i\leq k,\quad \sum_{i=1}^mr_i=k
	\biggr\}
\end{IEEEeqnarray}

Define the matrix $\mathbf M\in\mathbb{K}_n^{k\times k}$ in terms of the parameters $[(S_i,r_i)]_{i=1}^m\in\mathcal{S}_{k,m,n}$ as follows
(We will often write ${\mathbf M[(S_i,r_i)]_{i=1}^m}$ to indicate its parameters):
\begin{IEEEeqnarray}{rCl}\setlength\arraycolsep{3pt}
	\mathbf M=
	\begin{pmatrix}
		1  & {\displaystyle\sum_{j\in S_1}\!\alpha_j}& \dots & {\displaystyle\prod_{j\in S_1}\!\alpha_j} &0 &\cdots &&\\
		& \ddots& & & \ddots\\
		0 & & 1 & {\displaystyle\sum_{j\in S_1}\!\alpha_j} & \dots &  {\displaystyle\prod_{j\in S_1}\!\alpha_j}  &0&\cdots\\
		\hline
		&&\vdots\\
		\hline
		1  & {\displaystyle\sum_{j\in S_m}\!\alpha_j}& \dots & {\displaystyle\prod_{j\in S_m}\!\alpha_j} &0 &\cdots &&\\
		& \ddots& & & \ddots\\
		0 & & 1 & {\displaystyle\sum_{j\in S_m}\!\alpha_j} & \dots &  {\displaystyle\prod_{j\in S_m}\!\alpha_j}  &0&\cdots
	\end{pmatrix}
	\begin{matrix}
		\left.\vphantom{\begin{pmatrix}{\displaystyle\sum_{S_1}}\\\ddots\\{\displaystyle\sum_{S_1}}\end{pmatrix}}
		\right\} r_1\\
		\vdots\quad \\
		\left.\vphantom{\begin{pmatrix}{\displaystyle\sum_{S_1}}\\\ddots\\{\displaystyle\sum_{S_1}}\end{pmatrix}}
		\right\} r_m
	\end{matrix}
\end{IEEEeqnarray}
The rows are partitioned into $m$ blocks and for $i\in[m]$, the $i$th block is an $r_i\times k$ upper triangular Toeplitz matrix,
whose first row consists of the coefficients of the polynomial $x^{k-|S_i|-1}\prod_{j\in S_i}(x+\alpha_j)$ in descending order with respect to the degree. This matrix can be also thought of as a generalized Sylvester matrix that is constructed by $m$ polynomials.
The condition $|S_i|+r_i\leq k$ ensures that the rows are not shifted too much that we lose a nonzero entry in the last row of the $i$th block. Also, notice that the bottom-right entry of the $i$th block is nonzero if we have the equality $|S_i|+r_i=k$, and zero otherwise. So, we need at least one equality if we want $\mathbf M$ to be nonsingular.

In \Cref{proposition}, we give an equivalent way of writing $\det \mathbf M = 0$ in terms of the polynomials that we use  when constructing $\mathbf M$.
\begin{proposition}
	\label{proposition}
	Let $[(S_i,r_i)]_{i=1}^m\in\mathcal{S}_{k,m,n}$. For $i\in[m]$, define
	\begin{equation}
	p_i=x^{k-|S_i|-r_i}\prod_{j\in S_i}(x+\alpha_j)
	\end{equation}
	Then, $\det\mathbf M[(S_i,r_i)]_{i=1}^m=0$ if and only if there exist $q_1,\dots, q_m\in\mathbb{K}_n[x]$, not all zero, such that $\deg q_i\leq r_i-1$ for $i\in[m]$ and $\sum_{i=1}^mp_iq_i=0$.
	\hfill$\diamond$
\end{proposition}
\begin{proof}
	For each $q_i$, construct a row vector of size $r_i$ consisting of the coefficients of $x^{r_i-1},\dots,x,1$ in $q_i$
	and merge them into one row vector $\mathbf y\in\mathbb{K}_n^{1\times k}$.
	Then, ${\sum_{i=1}^mp_iq_i=0}$ iff  ${\mathbf y\cdot\mathbf M=0}$.
	The statement follows from $\det\mathbf M=0$ iff there exists nonzero ${\mathbf y\in\mathbb{K}_n^{1\times k}}$ such that ${ \mathbf y\cdot \mathbf M=0}$.
\end{proof}

\Cref{maintheorem} gives necessary and sufficient conditions on the parameters $[(S_i,r_i)]_{i=1}^m$ for $\det\mathbf M$ to be nonzero.
\begin{theorem}
	\label{maintheorem}
	Let $k\geq m\geq 1$, $n\geq 0$,
	${[(S_i,r_i)]_{i=1}^m\in\mathcal{S}_{k,m,n}}$.
	Then, $\det\mathbf M[(S_i,r_i)]_{i=1}^m \neq 0$
	if and only if
	for any nonempty $\Omega\subset[m]$,
	\begin{equation}
	\label{condition}
	\left|\bigcap_{i\in\Omega}S_i\right| + \sum_{i\in\Omega}r_i \leq \max_{i\in\Omega} |S_i|+r_i
	\end{equation}
	\hfill$\diamond$
\end{theorem}
\begin{proof}
	See Appendix \ref{proofmainthm}.
\end{proof}

In \Cref{maintheorem}, if we let $k=m$, $r_i=1$, and $|S_i|=k-1$, we will get \Cref{corollary1}.

\begin{corollary}
	\label{corollary1}
	Let $S_1,S_2,\dots,S_k\subset[n]$ such that $|S_i|=k-1$.
	Then, the determinant of
    \begin{IEEEeqnarray}{rCl}\setlength\arraycolsep{3pt}
    	\mathbf M[(S_i,1)]_{i=1}^k=
    	\begin{pmatrix}
    		1 & \sum_{j\in S_1}\alpha_j & \cdots & \prod{j\in S_1}\alpha_j\\
            1 & \sum_{j\in S_2}\alpha_j & \cdots & \prod{j\in S_2}\alpha_j\\
            \vdots & \vdots & & \vdots\\
            1 & \sum_{j\in S_m}\alpha_j & \cdots & \prod{j\in S_m}\alpha_j
    	\end{pmatrix}\quad
    \end{IEEEeqnarray}
    is nonzero if and only if
    for any nonempty $\Omega\subset[k]$,
	\begin{equation}
	\label{condcorollary}
	\left|\bigcap_{i\in\Omega}S_i\right|\leq k-|\Omega| 
	\end{equation}
	\hfill$\diamond$
\end{corollary}

%
%

%
%
\appendices
\section{Proof of \Cref{maintheorem}}\label{proofmainthm}
Suppose that for some nonempty $\Omega\subset[m]$, the condition (\ref{condition}) is not true.
Let $S_0=\bigcap_{i\in\Omega}S_i$, $r_0=\sum_{i\in \Omega}r_i$, ${k'=\max_{i\in\Omega}|S_i|+r_i}$.
Then, $|S_0|+r_0>k'$.
Consider the $r_0$ rows of $\mathbf M$ in the blocks indexed in $\Omega$.
They all have zeros in their last $k-k'$ entries.
Let ${\mathbf M_0\in\mathbb{K}_n^{r_0\times k'}}$ be the submatrix consisting of these rows without including the last $k-k'$ columns.
We will prove that $\mbox{rank}\,\mathbf M_0<r_0$, which implies $\det \mathbf M= 0$. 
Let $\mathbf W=((-\alpha_j)^{1-i})_{i\in[k'],j\in S_0}$ be $k'\times |S_0|$ Vandermonde matrix.
Then, $\mathbf M_0\cdot \mathbf W=0$ because the polynomials with the coefficients in the rows of $\mathbf M_0$ vanish at $-\alpha_j$ for $j\in S_0$. Hence,
\begin{equation}
\mbox{rank}\,\mathbf M_0\leq k'-\mbox{rank}\,\mathbf W\leq k'- \min\{k',|S_0|\}< r_0
\end{equation}
which proves the first direction.\\

For the other direction, we will apply induction on the parameters $(k,m,n)$ considered in the lexicographical order.
For ${m=1}$, $\mathcal{S}_{k,1,n}=\{[(\emptyset,k)]\}$ and ${\det\mathbf M[(\emptyset,k)]= \det\textbf I_k=1}$.
For $n=0$, all of $S_i$'s are empty; hence, for ${\Omega=[m]}$, (\ref{condition}) yields ${m=1}$, for which, we already showed ${\det\mathbf M=1}$.
For ${k\geq m\geq 2}$ and ${n\geq 1}$, assume that the statement is true for parameters $(k',m',n')$ that are smaller than $(k,m,n)$ with respect to lexicographical order.
Take any ${[(S_i,r_i)]_{i=1}^m\in\mathcal{S}_{k,m,n}}$ that satisfies the condition (\ref{condition}). We will prove that ${\det\mathbf M[(S_i,r_i)]_{i=1}^m\neq 0}$ under three cases:
\begin{enumerate}
	\item There exists $\Omega_1\subset[m]$ such that $2\leq|\Omega_1|\leq m-1$ and
	\begin{equation}
	\label{case1}
	\left|\bigcap_{i\in\Omega_1}S_i\right| + \sum_{i\in\Omega_1}r_i = \max_{i\in\Omega_1} |S_i|+r_i
	\end{equation}
	
	\item There exists a unique $i\in[m]$ such that $|S_i|+r_i=k$.
	
	\item Else (i.e. 1 and 2 are false).
\end{enumerate}

\subsection*{Case 1}
Let $\Omega_2=\{0\}\cup[m]-\Omega_1$.
Note that $2\leq |\Omega_1|,|\Omega_2|\leq m-1$.
Define
\begin{equation}
S_0 = \bigcap_{i\in\Omega_1}S_i
,\quad
r_0 = \sum_{i\in\Omega_1} r_i
\end{equation}
Then, (\ref{case1}) becomes
\begin{equation}
\label{case1becomes}
|S_0| + r_0 = \max_{i\in\Omega_1} |S_i|+r_i
\end{equation}
Define $S'_i=S_i-S_0$ for $i\in\Omega_1$.
Then,
\begin{equation}
[(S'_i,r_i)]_{i\in\Omega_1}\in\mathcal{S}_{r_0,|\Omega_1|,n}
,\quad
[(S_i,r_i)]_{i\in\Omega_2}\in\mathcal{S}_{k,|\Omega_2|,n}
\end{equation}
The first one is true because $r_0=\sum_{i\in\Omega_1}r_i$ and for any ${i\in\Omega_1}$, by (\ref{case1becomes}),
\begin{equation}
|S'_i|+r_i = |S_i|+r_i-|S_0| \leq r_0
\end{equation}
The second one is true because 
\begin{equation}
k=\sum_{i=1}^{m}r_i=\sum_{i\in\Omega_1}r_i+\sum_{i\in[m]-\Omega_1}r_i=\sum_{i\in\Omega_2}r_i,
\end{equation}
$|S_i|+r_i\leq k$ for $i\in [m]-\Omega_1$ and $|S_0|+r_0\leq k$ due to (\ref{case1becomes}).

By the induction hypothesis, the statement is true for
$[(S'_i,r_i)]_{i\in\Omega_1}$ and $[(S_i,r_i)]_{i\in\Omega_2}$.
We will show that both satisfy the condition (\ref{condition}):
\begin{enumerate}
	\item 
	For any nonempty $\Omega\subset\Omega_1$,
	\begin{align}
	\left|\bigcap_{i\in\Omega}S'_i\right| + \sum_{i\in\Omega}r_i
	&=\left|\bigcap_{i\in\Omega}S_i\right|-|S_0| + \sum_{i\in\Omega}r_i\\
	&\leq \max_{i\in\Omega} |S_i|+r_i - |S_0|\\
	&= \max_{i\in\Omega} |S'_i|+r_i
	\end{align}
	
	\item 
	For any nonempty $\Omega\subset\Omega_2$, if $0\notin\Omega_2$, then $\Omega\subset[m]$ and (\ref{condition}) holds trivially. Assume $\Omega=\{0\}\cup \Omega'$ for some $\Omega'\subset[m]-\Omega_1$. Then,
	\begin{align}
	\left|\bigcap_{i\in\Omega}S_i\right| + \sum_{i\in\Omega}r_i
	&=\left|\bigcap_{i\in \Omega_1\cup\Omega'}S_i\right| + \sum_{i\in\Omega_1\cup\Omega'}r_i\\
	&\leq \max_{i\in\Omega_1\cup\Omega'} |S_i|+r_i\\
	&=\max_{i\in\Omega} |S_i|+r_i
	\end{align}
\end{enumerate}
Hence, we have that
\begin{equation}
\det\mathbf M[(S'_i,r_i)]_{i\in\Omega_1}\neq 0
,\quad
\det\mathbf M[(S_i,r_i)]_{i\in\Omega_2}\neq 0
\end{equation}
Now, we will use \Cref{proposition}. Define for $i\in \{0\}\cup[m]$,
\begin{equation}
p_i=x^{k-|S_i|-r_i}\prod_{j\in S_i}(x+\alpha_j)
\end{equation}
and for $i\in\Omega_1$,
\begin{equation}
p'_i=x^{r_0-|S'_i|-r_i}\prod_{j\in S'_i}(x+\alpha_j)
\end{equation}
Note that for $i\in\Omega_1$, $p_i=p'_ip_0$.

Consider any $q_1,\dots,q_m\subset\mathbb{K}_n[x]$ such that $\deg q_i\leq r_i-1$ for $i\in[m]$ and $\sum_{i=1}^mp_iq_i=0$. We need to prove that $q_i=0$ for all $i\in[m]$. Define
\begin{equation}
q_0=\sum_{i\in\Omega_1}p'_iq_i
\end{equation}
Note that $\deg q_0\leq r_0-1$:
\begin{align}
\deg q_0
&\leq \max_{i\in \Omega_1}\, (\deg p'_i + \deg q_i)\\
&\leq \max_{i\in \Omega_1}\, ((r_0-r_i) + (r_i-1))\\
&= r_0-1
\end{align}
Also, we can write that
\begin{equation}
0=\sum_{i=1}^mp_iq_i = p_0\sum_{i\in\Omega_1}p'_iq_i + \sum_{i\in[m]-\Omega_1}p_iq_i = \sum_{i\in \Omega_2}p_iq_i
\end{equation}
Then, by \Cref{proposition}, we get $q_i=0$ for all $i\in\Omega_2$.
Then, $q_0=\sum_{i\in\Omega_1}p'_iq_i=0$.
Then, by \Cref{proposition}, $q_i=0$ for all $i\in\Omega_1$.
Hence, $q_i=0$ for all $i\in[m]$. By \Cref{proposition}, $\det\mathbf M[(S_i,r_i)]_{i=1}^m \neq 0$.

\subsection*{Case 2}
W.l.o.g., let $m$ be the maximizer. Then, for $i\in[m-1]$,
\begin{equation}
	k=|S_m|+r_m>|S_i|+r_i
\end{equation}
Then, the last column of $\mathbf M[(S_i,r_i)]_{i=1}^m$ is all zero except the last entry, which is $\prod_{j\in S_m}\alpha_j$.

Hence, we have
\begin{equation}
\det\mathbf M[(S_i,r_i)]_{i=1}^m = \det\mathbf M[(S_i,r'_i)]_{i=1}^m \cdot \prod_{j\in S_m}\alpha_j
\end{equation}
where $r'_m=r_m-1$ and $r'_i=r_i$ for $i\in[m-1]$ assuming that $r_m\geq 2$. (If $r_m=1$, the first multiplier would be $\det\mathbf M[(S_i,r_i)]_{i=1}^{m-1}$, which is nonzero by the induction hypothesis.)

Note that $[(S_i,r'_i)]_{i=1}^m\in\mathcal{S}_{k-1,m,n}$ since $\sum_{i=1}^mr'_i=k-1$ and $|S_i|+r'_i\leq k-1$ for any $i\in[m]$ due to the unique maximizer assumption.

By the induction hypothesis, the statement is true for $[(S_i,r'_i)]_{i=1}^m$.
If we prove that it satisfies the condition (\ref{condition}), then $\det\mathbf M[(S_i,r_i)]_{i=1}^m\neq 0$.

For any nonempty $\Omega\subset[m]$, if $m\notin\Omega$, then (\ref{condition}) holds trivially.
Assume $m\in\Omega$.
\begin{align}
\left|\bigcap_{i\in\Omega}S_i\right| + \sum_{i\in\Omega}r'_i
&=\left|\bigcap_{i\in\Omega}S_i\right| - 1+\sum_{i\in\Omega}r_i\\
&\leq \max_{i\in\Omega} |S_i|+r_i -1\\
&= k-1 \\
&= \max_{i\in\Omega} |S_i|+r'_i
\end{align}

\subsection*{Case 3}
For any nonempty $\Omega\subset[m]$ such that $|\Omega|\neq 1,m$, we have
\begin{equation}
\left|\bigcap_{i\in\Omega}S_i\right| + \sum_{i\in\Omega}r_i \leq \max_{i\in\Omega} |S_i|+r_i - 1
\end{equation}

Also, there exist at least two maximizers of $|S_i|+r_i$. W.l.o.g., assume that
\begin{equation}
k=|S_m|+r_m=|S_{m-1}|+r_{m-1}
\end{equation}

If $S_m= S_{m-1}$, we get a contradiction in (\ref{condition}):
\begin{equation}
r_m+r_{m-1} \leq \max\{r_m,r_{m-1}\}
\end{equation}

Then, either $S_{m-1}\neq [n]$ or $S_m\neq[n]$.
W.l.o.g., we can assume that $n\notin S_m$.
Substitute $\alpha_n=0$:
\begin{equation}
\left.\det\mathbf M[(S_i,r_i)]_{i=1}^m\right|_{\alpha_n=0}=\det\mathbf M[(S'_i,r_i)]_{i=1}^m
\end{equation}
where $S'_i=S_i-\{n\}$.

Note that $[(S'_i,r_i)]_{i=1}^m\in\mathcal{S}_{k,m,n-1}$ since $S'_i\subset[n-1]$ and $|S'_i|+r_i\leq |S_i|+r_i\leq k$  for $i\in[m]$.

By the induction hypothesis, the statement is true for $[(S'_i,r_i)]_{i=1}^m$.
If we prove that it satisfies the condition (\ref{condition}), then $\det\mathbf M[(S_i,r_i)]_{i=1}^m\neq 0$.

For $|\Omega|=1$, (\ref{condition}) holds trivially. For $|\Omega|\neq 1,m$, we have
\begin{align}
\left|\bigcap_{i\in\Omega}S'_i\right| + \sum_{i\in\Omega}r_i
&\leq \left|\bigcap_{i\in\Omega}S_i\right| + \sum_{i\in\Omega}r_i\\ 
&\leq \max_{i\in\Omega} |S_i|+r_i -1\\
&\leq \max_{i\in\Omega} |S'_i|+r_i
\end{align}
For $\Omega=[m]$, it is enough to show that ${k=\max_{i\in [m]}|S'_i|+r_i}$, which is true because
\begin{equation}
|S'_m|+r_m=|S_m|+r_m=k
\end{equation}


%
%



%
%


\end{document}